\documentclass{amsart}

\usepackage{amsmath,amsfonts}
\RequirePackage[OT1]{fontenc}
\RequirePackage[numbers]{natbib}
\RequirePackage[colorlinks,citecolor=blue,urlcolor=blue]{hyperref}

\usepackage{amsmath}    
\usepackage{amsfonts}   
\usepackage{amssymb}    
\usepackage{amsthm}     
\usepackage{mathrsfs}   
\usepackage{graphicx}   
\usepackage{color}
\usepackage{afterpage}

\numberwithin{equation}{section}
\theoremstyle{plain}
\newtheorem{theorem}{Theorem}[section]

\newtheorem{lemma}{Lemma}[section]

\newtheorem{corollary}{Corollary}[section]

\renewenvironment{proof}{\noindent{\it Proof.}}{\qed}
\newtheorem{proposition}{Proposition}[section]

\def\half{\hbox{$1\over2$}}
\def\qart{\hbox{$1\over4$}}
\def\bX{{\mathbb X}}

\newcommand{\BHV}{\textrm{BHV}}

\let\catsymbfont\mathcal

\newcommand{\aB}{{\catsymbfont{B}}}
\newcommand{\aC}{{\catsymbfont{C}}}

\begin{document}




\title{Testing to distinguish measures on metric spaces}
\vspace{1in}
\author{Andrew J. Blumberg}
\address{Department of Mathematics, University of Texas,
Austin, TX \ 78712}
\email{blumberg@math.utexas.edu}

\author{Prithwish Bhaumik}
\address{Data Scientist, Quantifid Inc., CA, 94025}
\email{prithwish1987@gmail.com}

\author{Stephen G. Walker}
\address{Department of Mathematics, University of Texas,
Austin, TX \ 78712}
\email{s.g.walker@math.utexas.edu}





\begin{abstract}
We study the problem of distinguishing between two distributions on a
metric space; i.e., given metric measure spaces $({\mathbb X}, d,
\mu_1)$ and $({\mathbb X}, d, \mu_2)$, we are interested in the problem
of determining from finite data whether or not $\mu_1$ is $\mu_2$.
The key is to use pairwise distances between observations and, employing a reconstruction
theorem of Gromov, we can perform such a test using a two sample Kolmogorov--Smirnov test.
A real analysis using phylogenetic trees and flu data is presented. 
\end{abstract}

\maketitle



\section{Introduction}
Statistical inference relies on the notion of observations  coming from some random mechanism taking the form of a probability model. In some instances the formulation of such a probability model is difficult and an illustration of this occurs when observations arise in the form of phylogenetic trees. To undertake statistical inference here one would need a probability model on the space of trees  with, say, $k$ leaves, and denoted by $\BHV_k$.  However, before one can contemplate how to proceed statistically it is necessary to determine a metric on this space.  A distance does exist which can be adequately used for statistical inference and this will be discussed later. For now we refer the reader to \cite{aldous1996}, \cite{aldous2001} and \cite{holmes2003} for background on statistical inference for phylogenetic trees. 

Much of statistical inference takes place in Euclidean spaces;
utilizing distances between parameters in $\mathbb{R}^d$, for
example. In such metric spaces many useful concepts such as centroids
and limit theorems exist. However, in other metric spaces, this is not
necessarily the case and in such circumstances statistical inference
becomes challenging. The particular task the present paper is
concerned with is two sample hypothesis testing. That is, given two
sets of randomly generated trees we wish to test that the two random
generating mechanisms are the same. This is different from the
hypothesis test considered in \cite{holmes2007}, who tested whether a
true phylogenetic tree, regarded as an unknown parameter, is included
in a specified set of trees or not. However, the theory we present
goes beyond this specific space of trees and applies to more general metric spaces. 

One such space, which we discuss only as a reference to provide illustration, arises
in nonparametric density estimation problems. One uses distances, or divergences, between density or distributon functions; such as the $L_1$ distance and Kullback--Leibler divergence for the former, and any distance metricizing weak convergence of distribution functions, for the latter. Though an atypical problem in statistics, one could ask whether two sets of randomly generated distribution functions, $F_1(t),\ldots,F_n(t)$, and $G_1(t),\ldots,G_n(t)$, share the same generating mechanism; i.e. to test $H_0:\mathcal{L}F=\mathcal{L}G$. For  a nonparametric test one is going to struggle to find a framework in which to conduct such a test. 

A parametric test could, for example, use the Dirichlet process (\cite{Ferguson1973}) as a foundation. The Dirichlet process measure for $F$ is fully specified by a scale parameter $c_F>0$ and a centering distribution $F_0$ for which 
$$\mathbb{E}F(t)=F_0(t)\quad\mbox{and}\quad\mbox{Var}F(t)=\frac{F_0(t)\{1-F_0(t)\}}{1+c_F}.$$
In this case we would look only to check whether the mean and variance of $F(t)$ matched those of $G(t)$ for each $t$. This would be however quite restrictive and miss other dissimilarities.

In order to undertake a nonparametric test we can utilize a metric $d_w$ metricizing weak convergence of distribution functions, such as the Prokhorov distance. From the observed sets of distribution functions we would compute the pairwise distances 
$$d_w(F_i,F_j)\quad\mbox{and}\quad d_w(G_i,G_j).$$
These $\half n(n-1)$ values from each set of distribution functions can then be used to construct a test. Indeed, \cite{gromov1987hyperbolic} has shown that if $M$ is the $n\times n$ matrix with entries $d_w(F_i,F_j)$ then the distribution of $M$ characterizes, up to an  isometry, the probability measure $\mu$ generating the $F$. The isometry, say $\tau$, would be such that $d_w(F_1,F_2)=d_w(\tau(F_1),\tau(F_2))$. This basic idea, which we adapt, can be used to perform a two sample test for $H_0:\mathcal{L}F=\mathcal{L}G$; though the caveat is that the power for testing $\mu$, which generates $F$, against $\mu_\tau$, which generates $G=\tau(F)$, will be no greater than the Type I error. However, we would not see this as a problem since the isometry alternatve is quite specific and unlikely to be present in practice; i.e. that $G$ is a specific deterministic transform of $F$.


The space we are primarily concerned with is the space of phylogenetic trees with $k$ leaves;
a non--Euclidean metric space.  The celebrated work of \cite{billera2001geometry}
constructs a metric space on $\BHV_k$.  The distance between two trees each with $k$ leaves is described in \cite{billera2001geometry}, and we refer the reader to this article for the details. 

While this space has a rich geometric structure, 
it is far from Euclidean.  
Performing statistical inference in such non--Euclidean metric spaces
is a problem of basic interest.  In this paper, we focus on the
problem of determining tests to distinguish finite samples $X$ and $Y$
drawn from different measures $\mu_1$ and $\mu_2$ on a metric space
$({\mathbb X},d)$.  We will assume that $\mu_1$ and $\mu_2$ are Borel
measures, ${\mathbb X}$ is a Polish space, and ${\mathbb X}$ is finite
with respect to $\mu_1$ and $\mu_2$, indeed $\mu_1(\mathbb{X})=\mu_2(\mathbb{X_2})=1$.  We approach the problem using
the intrinsic structure of the metric measure spaces $({\mathbb X}, d,
\mu_1)$ and $({\mathbb X}, d, \mu_2)$ following the geometric view  introduced by Gromov, \cite{gromov1987hyperbolic}.  Gromov's
``mm-reconstruction theorem'' shows that a metric measure space
$({\mathbb X}, d, \mu)$ is characterized up to isomorphism by the
induced distributions, via pushforward, under the
``distance matrix'' maps
\[
\underbrace{\bX \times \bX \times \ldots \times \bX}_{n} \to M_n,
\]
where $M_n$ denotes the set of positive symmetric $n \times n$
matrices and the map takes a set $(x_1,\ldots,x_n)$ of $n$ points to the matrix $M$
such that $M_{ij} = d(x_i, x_j)$.

Given two finite subsets $X, Y \subset {\mathbb X}$, we use distance matrix
distributions to specify and describe a non--parametric test for the hypothesis that $X$ and $Y$ were drawn from identical measures, up to an isometry, on ${\mathbb X}$.  
We validate our test using synthetic data, and also comparing with an alternative test, as well as using the test on antigenic flu data; coming from the National Centre for Biotechnology Information (NCBI).

The technique we use is based on that of Gromov. We obtain a modification of his result by characterizing a probability measure $\mu$ by the distribution of the stochastic process 
$$S(t)=\mu\left(x'\,:\, d(X,x')\leq t\right)\quad\mbox{with}\quad X\sim \mu.$$
With $n$ data points we can partially observe $n$ such processes and this will form the basis of the test.

An array of tests which could be employed are based on the so--called {\em Energy Statistics}; see \cite{szekely2013energy}. The energy distance between $X$ and $Y$ is defined as
$$\mathbb{D}(X,Y)=2\,{\mathbb E}(d(X,Y))-{\mathbb E}(d(X,X'))-{\mathbb E}(d(Y,Y'))$$
where $X'$ is an independent copy of $X$ and $Y'$ an independent copy of $Y$. Tests for $X=_d Y$ are therefore based on a sample estimate of $\mathbb{D}$ given by
$$\widehat{\mathbb{D}}=2\sum_{1\leq i,j\leq n}d(x_i,y_j)-\sum_{1\leq i,j\leq n}d(x_i,x_j)- \sum_{1\leq i,j\leq n}d(y_i,y_j).$$
This is using a fundamentally different strategy from the one we use.

There are at least three identified problems with using tests based on ${\mathbb D}(X,Y)$.  The first is that $(\bX,d)$ must be a Hilbert space (see \cite{lyons2013distance}) in order for the condition ${\mathbb D}(X,Y)=0\Longrightarrow X=_d Y$, which is obviously essential. Secondly, it has recently been shown that such tests do not have the sufficient power they were originally thought to have; see 
\cite{ramdas2015decreasing}. Finally, the two sample test requires boostrap and permutation techniques in order to implement; see Section 6 in \cite{szekely2013energy}.

Another idea  is to introduce a reference set $R$ and to consider a
test based on the values $d(r,x_i)$ and $d(r,y_i)$ for all $r\in
R$. See, for example, \cite{cheng2017}, though the authors are
primarily proposing a test on a high dimensional Euclidean space. Our
problem in the context of topological data would be how to choose $R$
in a meaningful way.

In section 2 we describe some theory which complements that of Gromov
and provides the basis for the test statistic. In Section 3 we illuminate the theory of Section 2 by showing how things look in a particular setting where we can identify objects quite easily. Section 4 presents
ilustrations including some real data analysis and section 5 concludes
with a brief discussion.

\section{Theory}

We adapt the basic technique of Gromov.  Note that Gromov's results
characterize the metric measure space only up to isomorphism; i.e.,
measure--preserving isometry.  In our setting, this minor
indeterminacy manifests itself via the fact that in order to test for
$\mu$ to be characterized by $S(\cdot)$, we must ensure the values of
$d(x,x')$ characterize the measure $\mu$; hence we assume that there
does not exist a measure-preserving isometry $\tau\colon
\bX\rightarrow \bX$.  For suppose there exists such a $\tau\colon
\bX\rightarrow \bX$ for which 
\begin{equation}
\label{condi}
\mu\left(s:d(x,s)\leq t\right)=\mu\left(s: d(\tau(x),\tau(s))\leq
t\right),
\end{equation}
for all $x$ and $t$.  Then based on values of $d(x,x')$, we cannot
infer $\mu$.

Since this hypothesis is unknowable in practice, we will view our
proposed test as testing for the hypothesis that the samples are drawn
from measures on $\bX$ which coincide after application of some such
$\tau$, notably including possibly the identity map.

Define $B_t(x)=\{x':d(x,x')\leq t\}$ and $\mu(B)$ is the mass assigned to the set $B$. Then we define the $[0,1]$ valued stochastic process $S_\mu(t)$, indexed by $t\geq 0$, as
$$S_\mu(t)=\mu(B_t(X))\quad\mbox{with}\quad X\sim \mu.$$
Hence, for example,
$${\mathbb E} \,S_\mu(t)=\int \mu(B_t(x))\,\mu(d x).$$

\begin{theorem}
The joint distribution of $(d(x,x_1),d(x,x_2),\ldots)$ with the $x,x_1,x_2,\ldots$ being i.i.d. from $\mu$ characterize the process $\{S_\mu(t)\}_{t\geq 0}$.
\end{theorem}

\begin{proof}
The moments
$$\mathbb{E}\prod_{j=1}^m S_\mu(t_j)^{r_j}=\int\ldots\int \prod_{j=1}^{r_1}1\{d(x,x_{j})\leq t_1\}\times\cdots\times$$
$$\prod_{j=r_1+\cdots+r_{m-1}+1}^{r}1\{d(x,x_{j})\leq t_m\}\,\mu(dx)\prod_{j=1}^r\mu(dx_j),$$
where $r=r_1+\cdots+r_m$, characterises the joint distribution of
$$(S_\mu(t_1),\ldots,S_\mu(t_m)),$$
for any choice of $(t_1,\ldots,t_m)$ under the measure $\mu$.
It is seen that this is characterized by the joint distribution of
$$D_\mu=\big(d(x,x_1),d(x,x_2),\ldots\big).$$
So this joint distribution characterizes the process $S_\mu(t)_{t\geq 0}$ through the finite dimensional distributions, which clearly satisfy the Kolmogorov consistency condition, thus completing the proof.
\end{proof}

\noindent
Hence, if $D_{\mu_1}=_d D_{\mu_2}$ for measures $\mu_1$ and $\mu_2$, i.e.
$$\big(d(x,x_1),d(x,x_2),\ldots\big)=_d \big(d(x',x_1'),d(x',x_2'),\ldots\big),$$
with $(x,x_1,x_2,\ldots)$ and $(x',x_1',x_2',\ldots)$ being i.i.d. from $\mu_1$ and $\mu_2$, respectively,  then
$$S_{\mu_1}(t)_{t\geq 0}=_d S_{\mu_2}(t)_{t\geq 0}.$$
That is, the two processes share the same distribution.

\begin{theorem}
\label{th:2}
Suppose the two processes $S_{\mu_1}(\cdot)$ and $S_{\mu_2}(\cdot)$ are equal in distribution and $\{\mu_j(B_t(x))\}_{x\in \bX}$ form a set of distinct paths. Also assume that $\mu_1$ and $\mu_2$ are distributions with densities $p_1$ and $p_2$, respectively. That is, for $j=1,2$.
$$\mu_j(B)=\int_B p_j(x)\,\lambda(d x)$$
for some measure $\lambda$.
If there is no non--trivial bijection $\tau:\bX\rightarrow \bX$ such that
$$\mu_1\left(s:d(x,s)\leq t\right)=\mu_1\left(s: d(\tau(x),\tau(s))\leq t\right)$$
for all $x$ and $t$, then $\mu_1(B)=\mu_2(B)$ for all metric balls $B$.
\end{theorem}
\begin{proof}
For ease of exposition we first assume that $\lambda$ is a discrete measure and that
$P_j\{x\}$ is the mass assigned to $x$ for measure $\mu_j$. Assume that $\mu_1\ne\mu_2$ in that there exist some metric ball $B$ such that $\mu_1(B)\ne \mu_2(B)$. 
Given that $\{\mu_1(B_t(x))\}$ form a set of distinct paths in $t$ for each $x$, and similarly for $\mu_2$,
there exists a non--trivial bijection $\tau:\bX\rightarrow \bX$ such that for all $x\in \bX$ and all $t\geq 0$,
$\mu_1(B_t(x))=\mu_2(B_t(\tau(x)))$ and $P_1\{x\}=P_2\{\tau(x)\}$. These statements follow from the assumption $S_{\mu_1}=_d S_{\mu_2}$.

Therefore, from the former of these statements we have
$$\sum_{d(x,s)\leq t} P_1\{s\}=\sum_{d(\tau(x),s)\leq t}P_2\{s\}$$
and the right side can be written as
\begin{equation}
\label{key}
\sum_{d(\tau(x),\tau(s))\leq t}P_2\{\tau(s)\}
\end{equation}
by a simple transformation.
Now using the latter of the statements, (\ref{key}) is given by
$$\sum_{d(\tau(x),\tau(s))\leq t} P_1\{s\}$$
and hence for all $t$,
$\mu_1\left(s:d(x,s)\leq t\right)=\mu_1\left(s: d(\tau(x),\tau(s))\leq t\right),$
contradicting the assumption in the statement of the theorem; hence,
for all metric balls $B$, $\mu_1(B)=\mu_2(B)$. 

The argument for other measures $\lambda$ follows similarly. For now we have the two results $\mu_1(B_t(x))=\mu_2(B_t(\tau(x)))$ and that if $x\sim\mu_1$ then $\tau(x)\sim \mu_2$. Hence, from the former result we obtain
$$\mathbb{E}_{s\sim\mu_1}{\bf 1}(d(x,s)\leq t)=\mathbb{E}_{s\sim\mu_2} {\bf 1}(d(\tau(x),s)\leq t)$$
and from the latter that
$$\mathbb{E}_{s\sim\mu_2} {\bf 1}(d(\tau(x),s)\leq t)=\mathbb{E}_{s\sim\mu_1} {\bf 1}(d(\tau(x),\tau(s)\leq t),$$
where ${\bf 1}$ denotes the usual indicator function.
\end{proof}

There is no asymmetry here in the use of $\mu_1$ in the statement of the thereom; for if no $\tau$ exists for $\mu_1$ and $S_{\mu_1}=_dS_{\mu_2}$, then no $\tau$ exists for $\mu_2$ either.

We now show that the conclusion of Theorem \ref{th:2} implies
$\mu_1\equiv \mu_2$.  The argument we give is standard; e.g.,
see~\cite{buet2016} for a more comprehensive discussion.
Throughout, we work with a fixed metric space $(\bX, d)$.  Let 
$\aC$ denote a collection of subsets of $\bX$.  For a given Borel
measure $\mu$ on $(\bX,d)$ we define an outer measure $\mu_C$
on $\bX$ by the formula
\[
\mu_C(A) = \inf_{\{C_i\} \subseteq \aC, A \subseteq \bigcup_i C_i} \sum_i \mu(C_i).
\]
Here we are considering only countable covers $\{C_i\}$ of $A$.

Let $\aB_{\epsilon}$ denote the collection of metric balls in $\bX$ of
radius $< \epsilon$.  We define a metric outer measure by the formula
\[
\mu^* (A) = \sup_{\epsilon \to 0} \mu_{\aB_{\epsilon}}(A).
\]
Note that for $\epsilon_1 < \epsilon_2$,
$\mu_{\aB_{\epsilon_1}}(A) \geq \mu_{\aB_{\epsilon_2}}(A)$.  Associated
to a metric outer measure is a Borel measure; we will abusively denote
this by $\mu^*$ as well.

Now, the hypothesis we are working with is that we have two measures
on $(\bX,d)$, $\mu_1$ and $\mu_2$, such that for any metric
ball $B$ we have the equality $\mu_1(B) = \mu_2(B)$.  Therefore, the
metric outer measures and associated Borel measures $\mu_1^*$ and
$\mu_2^*$ coincide.

We would like to specify conditions under which the original measures
$\mu_1$ and $\mu_2$ must coincide.  We will do this by providing
conditions under which $\mu_1 = \mu_1^*$ and $\mu_2 = \mu_2^*$; that
is, conditions under which a measure $\mu$ is determined by its values
on metric balls.

\begin{lemma}
For any metric measure space $(\bX, d, \mu)$ and subset
$A \subseteq \bX$, we have
\[
\mu(A) \leq \mu^*(A).
\]
\end{lemma}

\begin{proof}
This follows from the fact that for any fixed $\epsilon$ we can
approximate $\mu_{\aB_{\epsilon}}(A)$ arbitrarily closely by taking a
cover of $A$ by metric balls.
\end{proof}

In order to show that $\mu^*(A) \leq \mu(A)$, we need a hypothesis
on $\mu$.  All of the relevant hypotheses amount to control on the
approximation of an arbitrary set by metric balls, as one would
expect.  Here is a fairly common hypothesis that suffices: Recall that
a measure is doubling if there exists a finite constant $k > 0$ such
that $\mu(B_{2\epsilon}(x)) \leq k \mu(B_\epsilon(x))$ for all
$\epsilon > 0$ and $x \in \bX$.  We need the following standard result
about doubling measures.

\begin{proposition}
Let $(\bX,d, \mu)$ be a metric measure space with $\mu$ a
doubling measure.  For any open $A \in \bX$, there exists a countable
collection of pairwise disjoint balls $\{B_i\}$ such that $\bigcup_i
B_i \subseteq A$ and $\mu(A \backslash \bigcup_i B_i) = 0$.
\end{proposition}

We can now prove that when $(\bX, d, \mu)$ is a metric
measure space with $\mu$ a doubling measure,
$\mu^*(A) \leq \mu(A)$ for all $A \subseteq \bX$.  It suffices to
consider $A$ open.  Furthermore, it is clear that $\mu^*$ and
$\mu_{\aB_\epsilon}$ are also doubling measures if $\mu$ is, and so
for $\epsilon > 0$ we apply the proposition to $\mu_{\aB_\epsilon}$
and conclude that  
\[
\mu_{\aB_\epsilon}(A) = \mu_{\aB_\epsilon}\left(\bigcup_i B_i\right) \leq
\sum_i \mu(B_i) \leq \mu(A).
\]
The desired inequality now follows by letting $\epsilon$ go to zero.

\section{Theory on space of normal density functions}

In this section we provide an illustration of the theory in Section 2 so as to expose some of the key points in perhaps more familiar statistical territory. Though we need point out statistical testing in the environment of normal density functions to be described would not be done this way; we use it merely to illuminate Section 2. We let $\mathbb{X}$ be the non--Euclidean space of normal density functions and $d$ is the Hellinger distance. So $\mathbb{X}=\{\mbox{N}(\cdot|\mu,\sigma^2)\}$ and 
$$d^2((\mu_1,\sigma_1),(\mu_2,\sigma_2))=1-\sqrt{\frac{2\sigma_1\sigma_2}{\sigma_1^2+\sigma_2^2}}\,\exp\left\{-\qart \frac{(\mu_1-\mu_2)^2}{\sigma_1^2+\sigma_2^2}\right\}.$$
For further simplicity, let us assume the normal densities both have mean $0$, so 
$$d^2(\sigma_1,\sigma_2)=1-\sqrt{\frac{2\sigma_1\sigma_2}{\sigma_1^2+\sigma_2^2}}.$$
A distance preserving bijective isometry exists here; it is $\tau(\sigma)=1/\sigma$, and it is easy to see that
$d(\sigma_1,\sigma_2)=d(\tau(\sigma_1),\tau(\sigma_2))$.  Hence, our test in this case would not be able to distinguish between the densities $f(\sigma)$ and $g(\sigma)=f(\tau(\sigma))\,|\tau'(\sigma)|$.
To elaborate, from $f$ we would have distances $d(\sigma_i,\sigma_j)$ and, for example, the distances from $g$ would be $d(1/\sigma_i,1/\sigma_j)=d(\sigma_i,\sigma_j)$ so we would conclude they were from the same source.
However, we would not see this as a problem, given the very specific nature of the $f$ and $g$. 

To consider the stochastic process $S_\mu(t)$, let us consider $d(\sigma,x)<t$, with $0<t<1$. This becomes
$$1-\sqrt{\frac{2\sigma x}{\sigma^2+x^2}}<t^2\Rightarrow \frac{2\sigma x}{\sigma^2+x^2}>s^2$$
where $s=1-t^2$. Hence
$$\mu(B_{t}(x))=\mbox{P}_\sigma\left(x/s^2-x\sqrt{1/s^2-1}<\sigma<x/s^2+x\sqrt{1/s^2-1}\right).$$
If, for example, $P_\sigma$ is standard exponential, then
$$\mu(B_t(x))=2\,\exp\left(-x/s^2\right)\sinh\left(x\sqrt{1/s^2-1}\right)$$
and $S_\mu(t)$ is the random path as a function of $t$ from 0 to 1, and with the $x$ chosen from $P_\sigma$. 
Here we see, for example, that for different $x$, the paths $\mu(B_t(x))$ are different.

However, we note in this case that $S_\mu(t)$ does not characterize $\mu$ due to the isometry.

To motivate the test described in the paper we would be testing $P_{\sigma_1}\equiv P_{\sigma_2}$ where we observe $X_i=\mbox{N}(\cdot|0,\sigma_{1i}^2)$ and $Y_i=\mbox{N}(\cdot|0,\sigma_{2i}^2)$ for $i=1,\ldots,n$, with $\sigma_{1i}\sim_{iid}P_{\sigma_1}$ and $\sigma_{2i}\sim_{iid}P_{\sigma_2}$. If we knew that the Fr\'echet means were both, say $\mbox{N}(\cdot|0,\sigma_0^2)$, with $\sigma_0$ known, then we would compare the samples
$$\left(d(\mbox{N}(\cdot|0,\sigma_0^2),X_i)\right)_{i=1}^n  \quad\mbox{and}\quad \left(d(\mbox{N}(\cdot|0,\sigma_0^2),Y_i)\right)_{i=1}^n$$
and use a Kolmogorov--Smirnov two sample test. This makes sense because
$$d(\mbox{N}(\cdot|0,\sigma_0^2),X)=_d d(\mbox{N}(\cdot|0,\sigma_0^2),Y)$$
is equivalent to
$$\frac{2\sigma_0\sigma_1}{\sigma_0^2+\sigma_1^2}=_d \frac{2\sigma_0\sigma_2}{\sigma_0^2+\sigma_2^2}.$$
This implies that
$$P(\sigma_1\in B(u))=P(\sigma_2\in B(u))\quad \mbox{for all}\quad u\in(0,1)$$
where
$$B(u)=\left(\sigma_0/u-\sigma_0\sqrt{1/u^2-1},\sigma_0/u+\sigma_0\sqrt{1/u^2-1}\right),$$
which implies $\sigma_1=_d\sigma_2$. This does not hold, for we can have $\sigma_2=1/\sigma_1$, when $\sigma_0=1$. That is, under this special existence of the known Fr\'echet mean, the appearance of the isometry arises when $\sigma_0=1$ only.

Note here that the Fr\'echet mean is given by the $\sigma_0$ which minimizes
$$\int d(\mbox{N}(\cdot|0,\sigma_0^2),\mbox{N}(\cdot|0,\sigma^2))\,dP(\sigma).$$
If a Fr\'echet mean does exist and is unique for each sample, but is unknown, then it can be estimated from the data. Indeed, we estimate the mean from the $X$ sample as $X_{i_X}$ which minimizes, over $i=1,\ldots,n$, 
$$\sum_{j\ne i}d (X_i,X_j).$$
A similar strategy is used to get the Fr\'echet mean for the $Y$ sample. The validity of a Kolmogorov--Smirnov test with a reduced degree of freedom, and with samples, 
$$(d(X_{i_X},X_j))_{j\ne i_X}\quad\mbox{and}\quad (d(Y_{i_Y},Y_j))_{j\ne i_Y},$$
is now provided by the theory in Section 2. The reasoning is that up to an isometry, the rows from $M_X$; i.e. $(d(X_i,X_j))_{j\ne i}$ all characterize $P_{\sigma_1}$ asymptotically. We use the minimum sum row in order to replicate as close as possible the case with the known Fr\'echet mean.

The samples being compared are
$$\left(\frac{\sigma_{1\,i_X}\sigma_{1j}}{\sigma_{1\,i_X}^2+\sigma_{1j}^2}\right)_{j\ne i_X}\quad\mbox{and}\quad 
\left(\frac{\sigma_{2\,i_Y}\sigma_{2j}}{\sigma_{2\,i_Y}^2+\sigma_{2j}^2}\right)_{j\ne i_Y}.$$
The first set yields an empirical distribution, say $F_X$, and the second set an empirical distribution, say $F_Y$. We then test for $F_X\equiv F_Y$ using the two sample Kolmogorov--Smirnov test.

Aside from the isometry, $\sigma_2=1/\sigma_1$, if these two samples pass the two sample Kolmogorov--Smirnov test with $n-1$ samples, then we do not reject the hypothesis $P_{\sigma_1}=P_{\sigma_2}$. 

\section{Test and Illustrations}

Suppose we have two samples $(X_1,X_2,\ldots,X_n)$ and $(Y_1,Y_2,\ldots,Y_n)$ from $\mu_1$ and $\mu_2$, respectively. Now if $\mu_1$ and $\mu_2$ are identical, then for each $i$ and $j$ , the 
distributions of 
$$d(X_i,X_k)_{k\ne i}\quad\mbox{and}\quad d(Y_i,Y_k)_{k\ne i}$$
are also identical. Moreover, from the theory in Section 2, we have that if the asymptotic sequences share the same distribution then $\mu_1$ and $\mu_2$ are the same. Hence, our proposed test is to select the two sets of values using carefully chosen $i=i_X$ and $i=i_Y$ and to then perform a two sample Kolmogorov--Smirnov test; see for example \cite{corder2014nonparametric}. Our choices for $i_X$ and $i_Y$ are
$$i_X=\arg\min_i \sum_{k=1}^n d(X_i,X_k) \quad\mbox{and}\quad i_Y=\arg\min_i \sum_{k=1}^n d(Y_i,Y_k), $$
respectively. This ensures we are comparing like rows from each matrix $M_X$ and $M_Y$ and moreover the choice $X_{i_X}$ and $Y_{i_Y}$ can be seen as a best representative of a location for the measures in that they are similar to Fr\'echet means. More on this in the discussion in Section 5.

\subsection{Comparison with synthetic data}

To compare  our test with another, we use a test statistic from the distance based test given in \cite{szekely2013energy}, and known as the energy test. The test statistic is given by
$$\widehat{\mathbb{D}}=n^{-2}\,\left[2\sum_{1\leq i,j\leq n}d(X_i,Y_j)-\sum_{1\leq i,j\leq n}d(X_i,X_j)-\sum_{1\leq i,j\leq n}d(Y_i,Y_j)\right].$$
This is currently a popular choice for comparing two data sets. However, there is a problem in that the critical value will depend on the distributions of $X$ and $Y$ and hence a bootstrap procedure is required in order to  
complete the test. 

To implement the two sample energy test we use bootstrap methods; since the distribution of the null hypothesis that $X$ and $Y$ come from the same source actually depends on that source. Hence, to get a critical value we compute $\widehat{\mathbb{D}}^{(b)}$, which is obtained by  randomly splitting the vector $(X,Y)$ into two equally sized sets and computing the energy statistic based on such a partition. The 95\% quantile value from the $(\widehat{\mathbb{D}}^{(1)},\ldots,\widehat{\mathbb{D}}^{(B)})$, for a large $B$, serves as the critical value, written as $c$. We then compute $\widehat{\mathbb{D}}$ and reject the hypothesis for $X$ and $Y$ coming from the same source if $\widehat{\mathbb{D}}>c$.  

We compare this with the conditional Kolmogorov--Smirnov test.  Hence, we use the test statistic
$$T=\sup_s |F_{X}(s)-F_Y(s)|$$
where $F_X$ is the empirical distribution of the $(d(X_{i_X},X_i))_{i\ne i_X}$ and $F_Y$ the corresponding empirical distribution function of the $(d(Y_{i_Y},Y_i))_{i\ne i_Y}.$ Here we reject the null hypothesis at the 95\% level of significance if $T>1.36\sqrt{2/(n-1)}$, in keeping with the Kolmogorov--Smirnov two--sample test theory.

\begin{center}
\begin{table}
\begin{center}
\begin{tabular}{ccc}
\hline
$\sigma^2$  & $T$-test power & $\widehat{\mathbb{D}}$--test  power \\\hline
1.2  & 0.06 &  0.06 \\
1.4  & 0.20 &  0.15 \\
1.6  &  0.42&  0.33 \\
1.8  &  0.61&  0.43 \\
2.0  &  0.81 &  0.70 \\
2.2  &  0.90&  0.84 \\
2.4  &  0.95 &  0.92 \\\hline
\end{tabular}
\caption{Power comparison for $T$-tests and $\widehat{\mathbb{D}}$--tests}
\end{center}
\end{table}
\end{center}

To undertake an initial simulation study we take the $(X_i)$ as independent and identically distributed from the standard normal distribution and take the $(Y_i)$ as independent normal with mean 0 and variance $\sigma^2$. We then, in the case of the energy test,  compute the power of the test for a range of $\sigma^2=(1.2,1.4,\ldots,2.4)$ and the results are reported in Table 1.

To complete the settings for the comparison, we took $B=1000$ and the number of simulations to record the power value was based on a Monte Carlo sample size of 1000. The sample size $n$ was $n=40$. The distance $d(x,y)$ employed is the absolute value between $x$ and $y$. While the observations are simple it is the distribution of $d(x,y)$ which matters and even if the $x$ are generated according to some highly complex and high--dimensional setting, the distribution of the $d(x,y)$ may yet be not unusual. 


\subsection{Real data analysis}

A real data analysis is now presented.  In the
paper~\cite{zairis2016}, the authors used hemagglutinin, an antigenic
surface glycoprotein, coding sequences in RNA from $1089$ flu samples
collected in the United States between 1993 and 2016.  These samples
were originally obtained from the GI--SAID EpiFlu database and then
processed.  Small phylogenetic trees with three, four, and five
leaves were constructed from randomly drawn sequential samples; i.e.,
representing samples from three, four, or five consecutive years, the
result was empirical distributions in the BHV metric space of
phylogenetic trees; see \cite{billera2001geometry}.  A phylogenetic
tree in $\BHV_n$ is an acyclic connected graph with a distinguished
vertex (the root) and $n$ vertices of degree $1$, the leaves, along
with labels in $\{1,n\}$ for the leaves and weights in ${\mathbb  R}_{\geq 0}$
for the edges that are not incident to a leaf.  For a given tree
topology, a tree is then completely specified by a vector in $R^{n-2}$
with all coordinates non--negative.  An illustration of a tree with 3 leaves is given in Fig~\ref{fig3}.

\begin{figure}[!htbp]
\begin{center}
\includegraphics[scale=0.4]{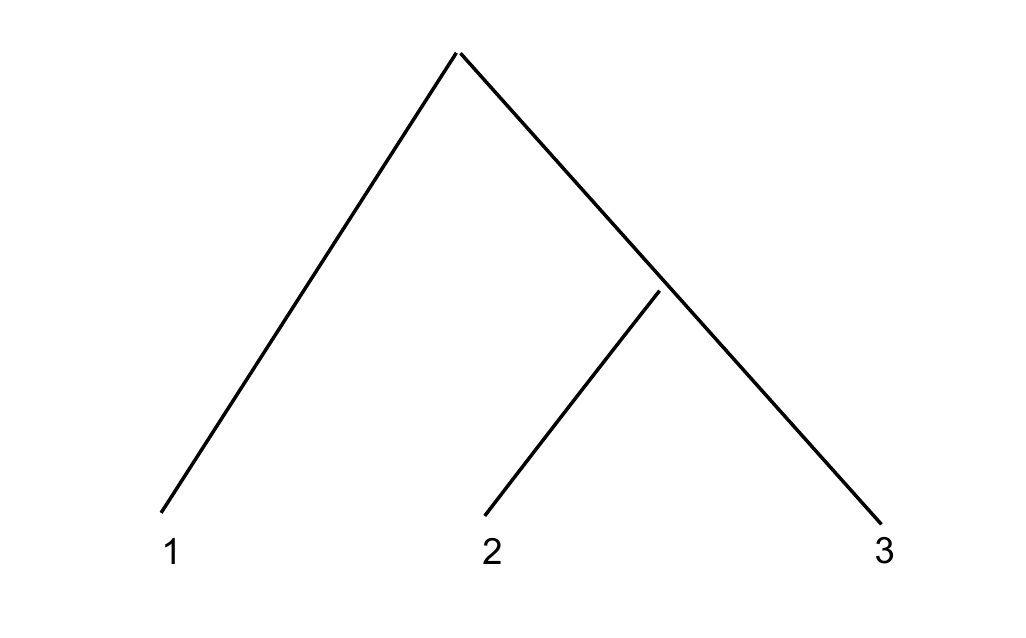}
\caption{Phylogenetic tree with three leaves}
\label{fig3}

\end{center}
\end{figure}

Note that switching labels 2 and 3 results in the same tree; we do not
take account of the embedding of the tree in the plane.  More general
illustrations of trees with the same and distinct topologies are given
in Fig~\ref{fig:tree_examples}.

\begin{figure}[!htbp]
\begin{center}
\includegraphics[scale=0.15]{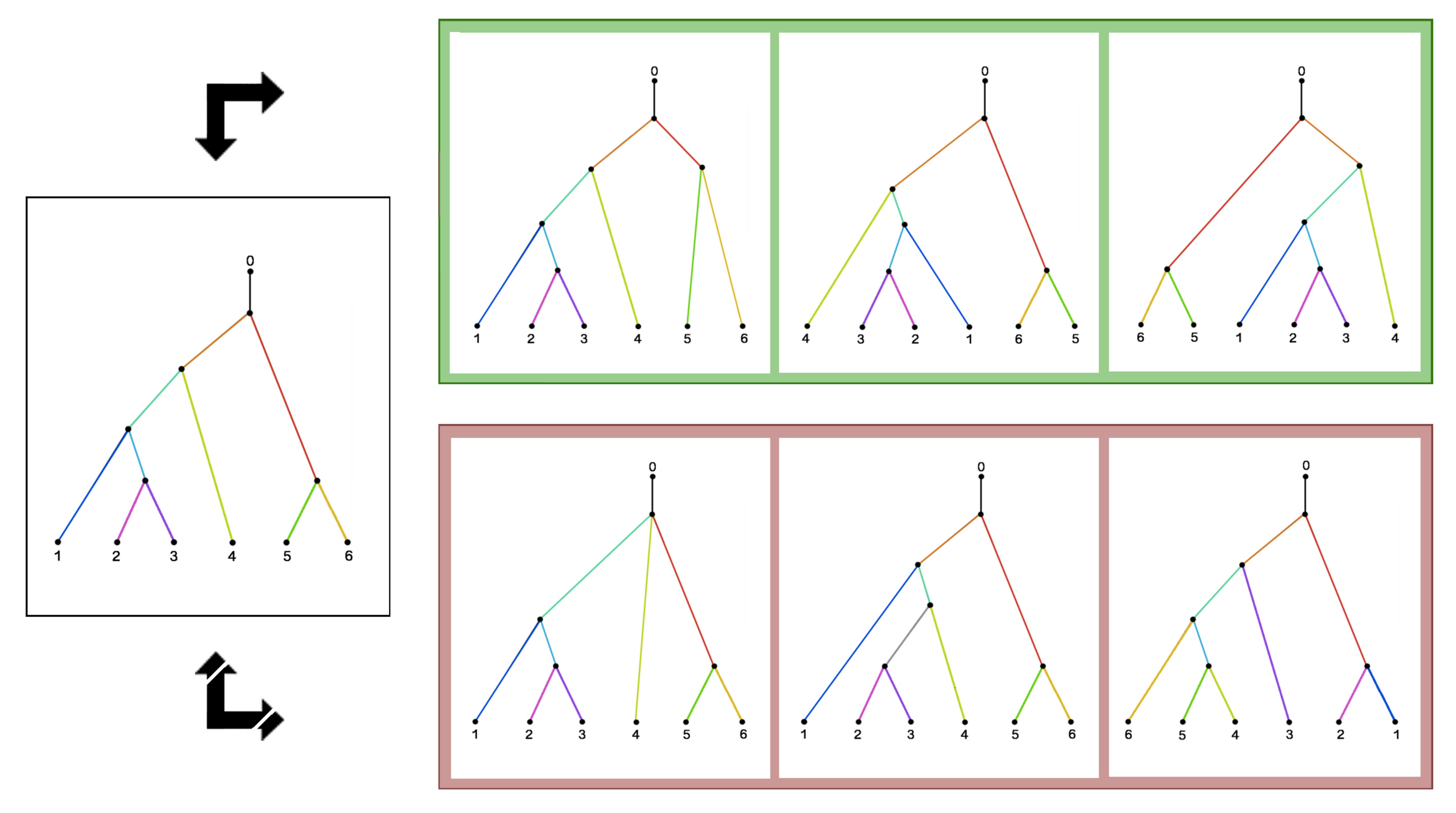}
\caption{Phylogenetic trees; top row is the same topology, bottom row
  has different topologies.}
\label{fig:tree_examples}

\end{center}
\end{figure}

To form the metric space $\BHV_n$,
we glue together these Euclidean orthants labelled by tree topologies
so that two orthants are adjacent if the two topologies coincide after
collapsing a single edge to $0$ in each one.  Such a transformation
is often referred to as a tree rotation.  

\begin{figure}[!htbp]
\begin{center}
\includegraphics[scale=0.2]{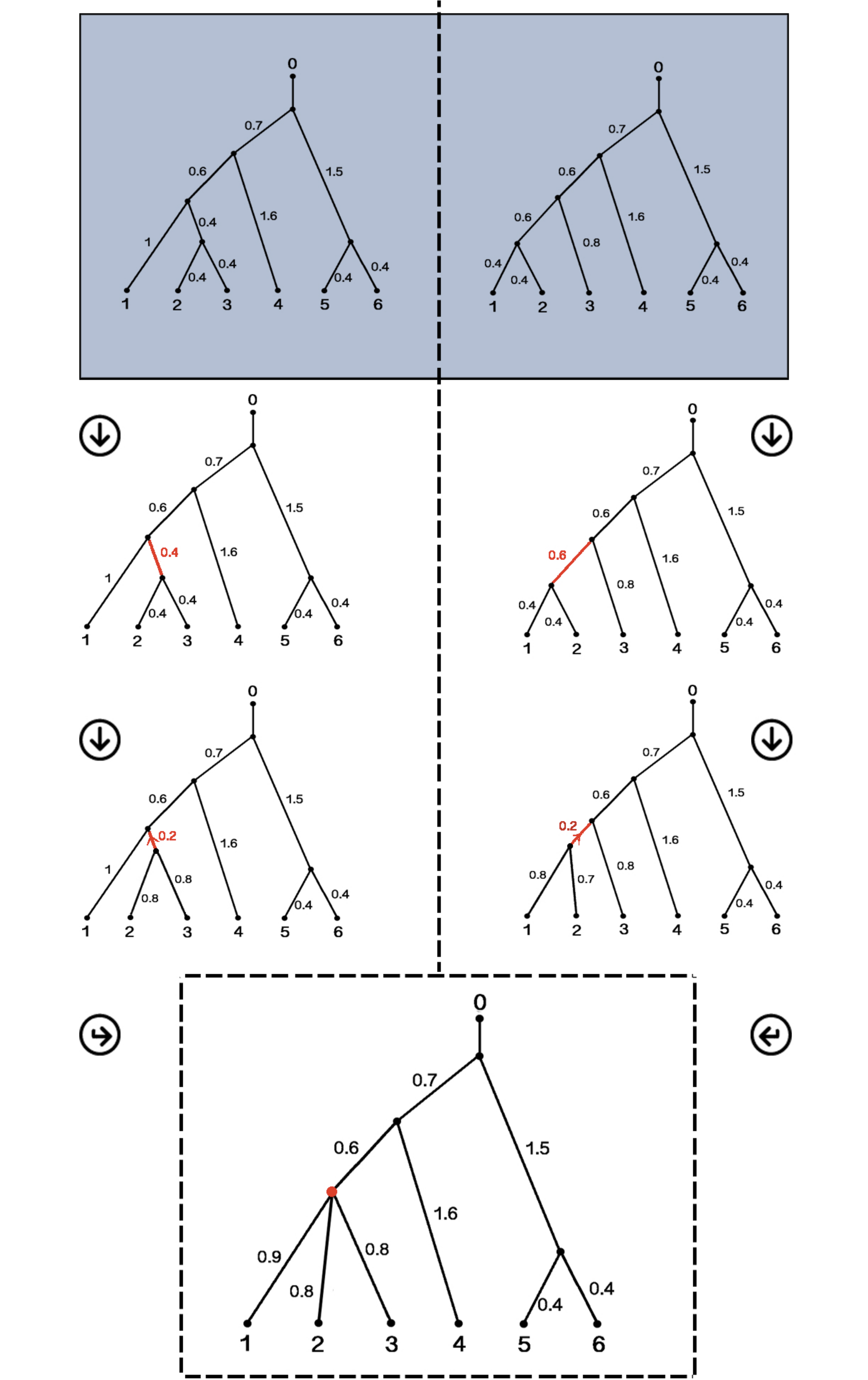}
\caption{Tree rotation}
\label{fig:tree_rot}

\end{center}
\end{figure}

The metric is now computed as the minimal piecewise linear path
between trees; see Fig~\ref{fig:tree_rot} for an example of a path
across orthants in tree space.  Note however that many paths go
through the ``cone point'', which is the tree with all internal edges
of length $0$.
As another illustration of the
distance, consider the tree from Fig~\ref{fig3} and the tree
obtained by switching 1 and 2 (which is a distinct tree).
To compute the distance we move the node for 2 and 3 in Fig~\ref{fig3}
to the root by collapsing the edge between it and the root, so that now
all the leaves hang from the same point, and then rearrange the leaves
appropriately.  The distance needed to move the nodes to the
root and then out again constitutes the overall distance. In
more complicated trees, the distance is harder to visualize;
nonetheless, there are (polynomial-time) computer algorithms to
compute it~\cite{Owen}.

In order to interpret the results of our test, we need to understand
the isometries of the space of phylogenetic trees.
In~\cite{Grindstaff}, it is shown that the only isometries $\BHV_n \to
\BHV_n$ are given by permuting labels; a permutation on the labels for
the leaves might transform the topology of a tree but clearly
preserves distances between trees.  As before, we do not see this as a
problem as we would not anticipate the two generating mechanisms for
the two sets of trees to be separated by such a permutation.

\begin{figure}[!htbp]
\begin{center}
\includegraphics[scale=0.55]{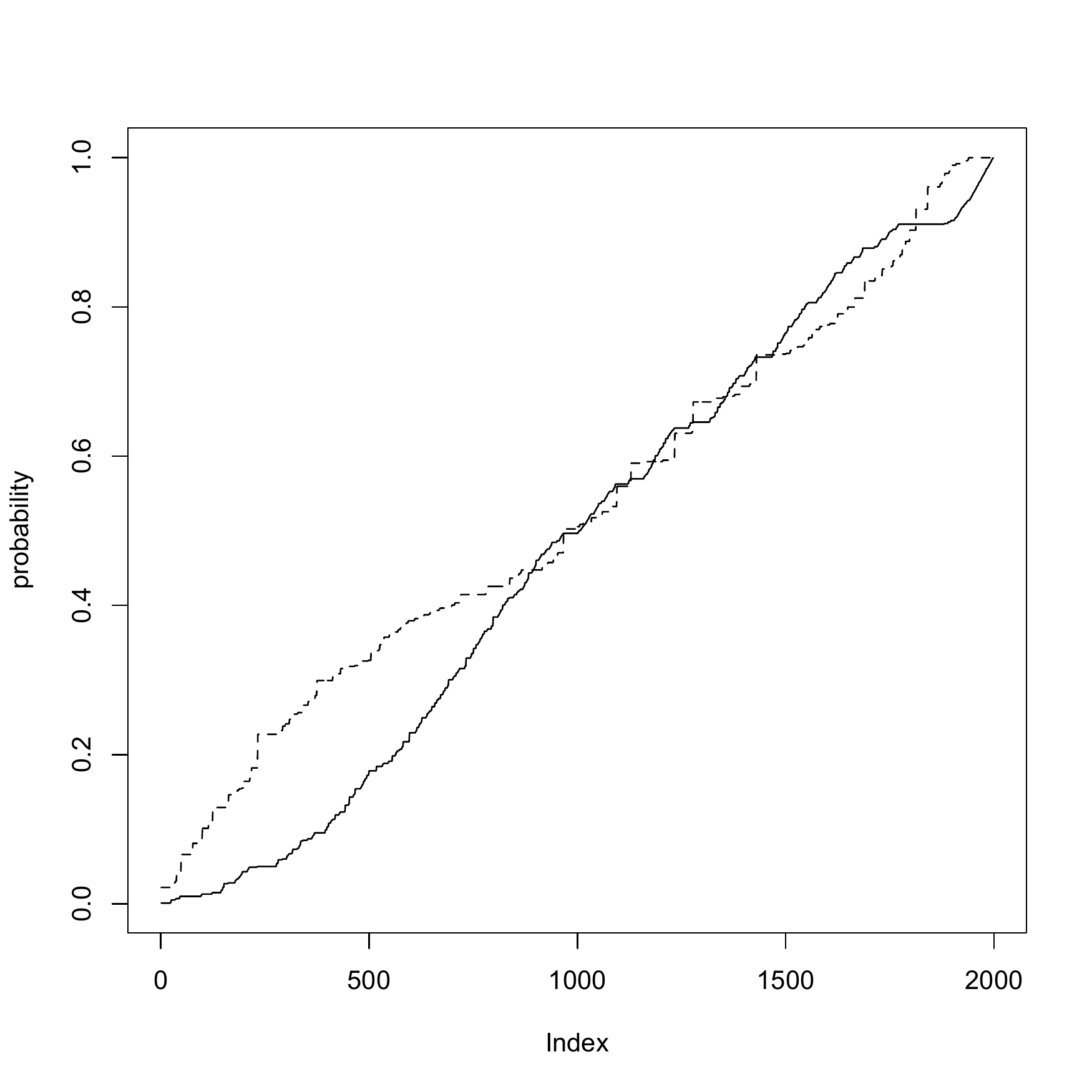}
\caption{Distribution functions of samples $d(X_{i_X},X_i)_{i\ne i_X}$ and $d(Y_{i_Y},Y_i)_{i\ne i_Y}$ }
\label{fig1}

\end{center}
\end{figure}

\begin{figure}[!htbp]
\begin{center}
\includegraphics[scale=0.55]{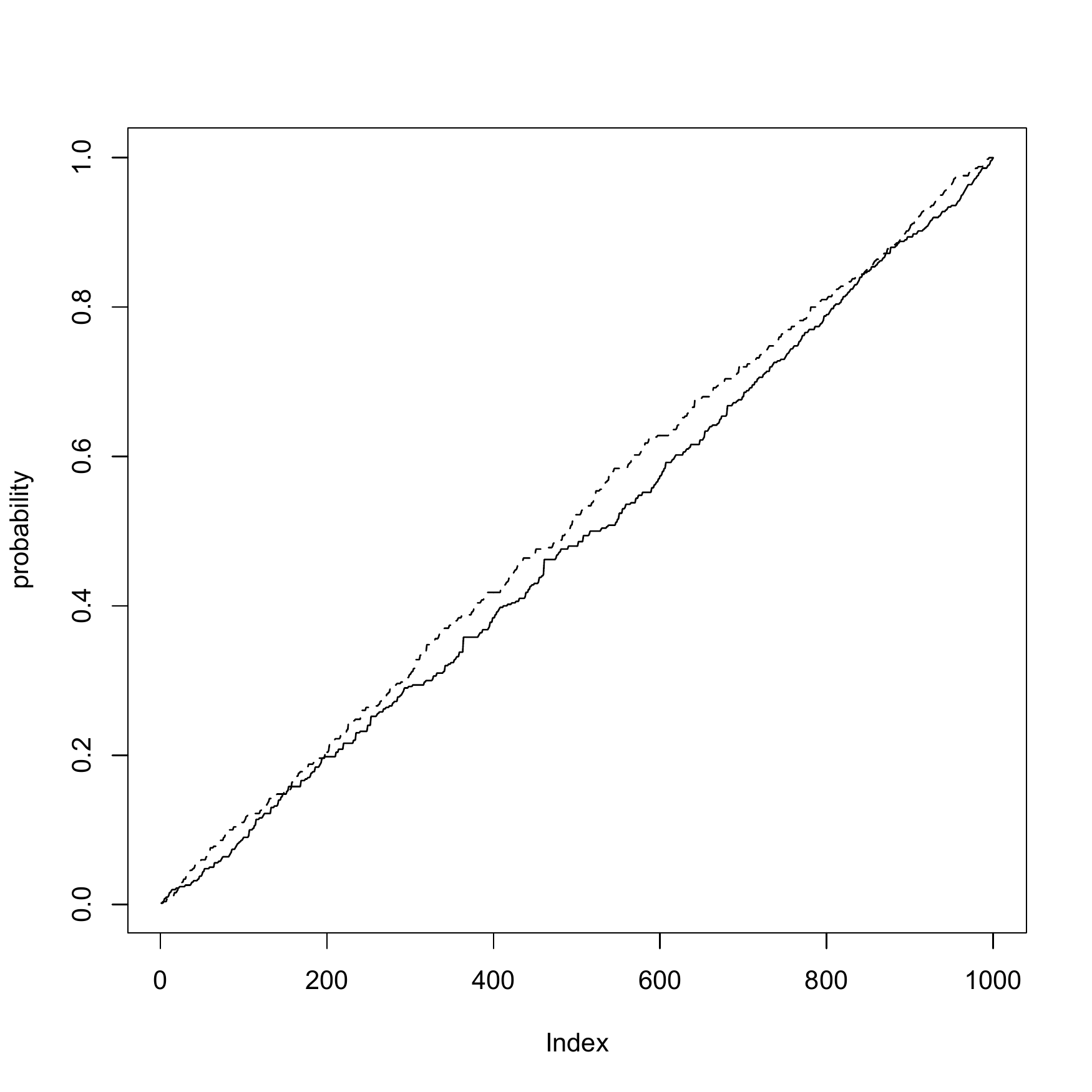}
\caption{Distribution functions of samples $d(X_{i_X},X_i)_{i\ne i_X}$, $1\leq i,i_X\leq n/2$, and $d(X_{i_X},X_i)_{i\ne i_X}$, $n/2 < i,i_X\leq n$}
\label{fig2}

\end{center}
\end{figure}

The purpose of looking at these distributions in the $\BHV$ metric
spaces was to use the metric to compare distributions for different
seasons; reliable ways to do this would lead to predictions about vaccine
design and effectiveness.  In~\cite{zairis2016}, a crude test
statistic involving the distance to the centroid turned out to be
surprisingly effective for predicting vaccine efficacy.  As a
representative experiment to validate our distribution test, we compared
empirical distributions computed from windows centered around the 1996
and the 2007 seasons, using data supplied by the authors
of~\cite{zairis2016}.  The metric distance was computed using software
implementing the fast algorithm of~\cite{Owen}.  We used 1000 samples from each distribution.

Both tests; i.e. the $T$ and $\widehat{\mathbb{D}}$ tests rejected, at a 95\% level of significance, the null hypothesis that the two distributions generating $X$ and $Y$ are the same. The Kolmogorov--Smirnov test is remarkably quick to implement.  However, for the two sample energy test there is an extremely time consuming computation of the critical value. To confirm the distribution of the distances are regular looking,  we plot the distribution functions of $d(X_{i_X},X_i)_{j\ne i_X}$ and $d(Y_{i_Y},Y_i)_{i\ne i_Y}$, with $i_X$ and $i_Y$ chosen as described previously. The two together are shown in Fig~\ref{fig1}.


Given the large number of data, we split the data matrix for the $X$, the 1996 flu data, into two independent blocks forming data matrix $M_{X}'$ based on real data $(d(X_{i},X_j))_{1\leq i,j \leq n/2}$ and data matrix $M_{X}''$ based on real data $(d(X_{i},X_j))_{n/2 < i,j\leq n}$. We now perform the  Kolmogorov--Smirnov test for these two matrices and in this case we accept the null hypothesis. We illustrate the two distribution functions in Fig~\ref{fig2}.

\section{Discussion}  

Hypothesis testing in the context of topological, non--Euclidean, data
poses unique challenges.  Even the problem of testing whether two sets
of data have the same source is difficult.  Our solution
to this is based on theory presented by Gromov which allows us to use
the Kolmogorov--Smirnov test on two carefully chosen sets of distances;
being those from the minimum row sums of the reconstruction matrix.  

On the other hand, an alternative test based on the energy distance is
also possible. Problems with this are that in the two--sample test one
has to compute the critical value using permutations and with large data sets this procedure can become
prohibitively time consuming.

We have demonstrated that our test is substantially faster and
moreover when there is a simple scale difference between the two
samples, it has superior power to that of the energy test. 

For further insights into the test, assume that $X$ and $Y$ have the same Fr\'echet mean $\theta$; then we can test for $\mu_X=\mu_Y$ using samples 
$d(\theta,X_i)$ and $d(\theta,Y_i).$
A suitable test would be the Kolmogorov--Smirnov test with $n$ points.

Ruling out any isometries, we can assume that 
$d(\theta,X)=_d d(\theta,Y)\quad\iff\quad X=_d Y.$
Not knowing $\theta$ we would use the empirical samples
$d(X_{i_X},X_i)_{i\ne i_X}$ and $d(Y_{i_Y},Y_i)_{i\ne i_Y}$
where $i_X$ minimizes over $j$ the sum
$\sum_{i\ne j} d(X_j,X_i),$
and similarly for $i_Y$. This can also be tested using the Kolmogorov--Smirnov test with now a reduction to $n-1$ points. This is of course the test we use.
The theory attributable to Gromov says this is also the test of choice when $\theta$ may not exist and the point of the use of empirical samples $X_{i_X}$ and $Y_{i_Y }$ replacing $\theta$   allows for the rejection of the hypothesis when the means, if they exist, are different. 

Another  interesting example also arises in the context of topological data
analysis; the space of ``barcodes'', which is the output of the
determination of persistent homology (e.g.,
see~\cite{weinberger2011persistent, ghrist2008}) forms a
metric space which is not Euclidean. 

Finally we mention that there is no reason why we can not handle $n$ and $m$ sample sizes from each measure; we simply chose to illustrate the test when sample sizes are equal.

\section*{Acknowledgements} 
The first author was supported in part by AFOSR grant FA9550-15-1-0302
and NIH grants 5U54CA193313 and GG010211-R01-HIV, and the third author
from NSF grants DMS 1506879 and 1612891.  The first author would like
to thank Raul Rabadan, Sakellarios Zairis, and Hossein Khiabanian for
helpful conversations.

\bibliographystyle{plain}
\bibliography{ref}

\newpage

\newpage

\end{document}